\documentclass[11pt]{amsart}
\usepackage{amssymb}
\usepackage{amsmath}
\usepackage{amscd}
\usepackage{amsfonts}
\usepackage{mathrsfs}
\usepackage{stmaryrd}
\usepackage{bbm}
\newtheorem{theorem}{Theorem}[section]
\newtheorem{prop}[theorem]{Proposition}
\newtheorem{defn}[theorem]{Definition}
\newtheorem{lemma}[theorem]{Lemma}
\newtheorem{coro}[theorem]{Corollary}
\newtheorem{prop-def}{Proposition-Definition}[section]

\newtheorem{remark}[theorem]{Remark}

\newtheorem{exam}[theorem]{Example}

\begin{document}
\setlength{\oddsidemargin}{0cm} \setlength{\evensidemargin}{0cm}

\title{Some results on L-dendriform algebras}

\author{Chengming Bai}
%$^*$}

\address{Chern Institute of Mathematics \& LPMC, Nankai University,
Tianjin 300071, P.R. China} \email{baicm@nankai.edu.cn}
%\thanks{$^*$
%Corresponding author}

\author{Ligong Liu}

\address{Chern Institute of Mathematics \& LPMC, Nankai
University, Tianjin 300071, P.R.
China}\email{liuligong@mail.nankai.edu.cn}

\author{Xiang Ni}

\address{Chern Institute of Mathematics \& LPMC, Nankai
University, Tianjin 300071, P.R.
China}\email{xiangn$_-$math@yahoo.cn}

\def\shorttitle{Some results on L-dendriform algebras}

\begin{abstract}

We introduce a notion of L-dendriform algebra due to several
different motivations. L-dendriform algebras are regarded as the
underlying algebraic structures of pseudo-Hessian structures on Lie
groups and the algebraic structures behind the $\mathcal
O$-operators of pre-Lie algebras and the related $S$-equation. As a
direct consequence, they provide some explicit solutions of
$S$-equation in certain pre-Lie algebras constructed from
L-dendriform algebras. They also fit into a bigger framework as Lie
algebraic analogues of dendriform algebras. Moreover, we introduce a
notion of $\mathcal O$-operator of an L-dendriform algebra which
gives an algebraic equation regarded as an analogue of the classical
Yang-Baxter equation in a Lie algebra.

\end{abstract}



\maketitle

%\tableofcontents \setcounter{section}{0}

\baselineskip=15.8pt

\section{Introduction}
\setcounter{equation}{0}
\renewcommand{\theequation}
{1.\arabic{equation}}

\subsection{Motivations}

In this paper, we introduce a new class of algebras, namely,
L-dendriform algebras, due to the following different motivations.

(1) {\bf Pseudo-Hessian structures.}\quad An L-dendriform algebra is
regarded as the underlying algebraic structure of a pseudo-Hessian
structure on a Lie group. In geometry, a Hessian manifold $M$ is a
flat affine manifold provided with a Hessian metric $g$, that is,
$g$ is a Riemannian metric such that for any each point $p\in M$
there exists a $C^\infty$-function $\varphi$ defined on a
neighborhood of $p$ such that
$g_{ij}=\frac{\partial^2\varphi}{\partial x^i\partial x^j}$. The
algebraic structure corresponding to an affine Lie group $G$ with a
$G$-invariant Hessian metric is a real pre-Lie algebra (see a survey
article \cite{Bu} for the study of pre-Lie algebras) with a
symmetric and positive definite 2-cocycle (\cite{Sh}). The Hessian
structures could be extended to pseudo-cases by replacing ``positive
definite" by ``nondegenerate" over the real number field, which
could be extended to the other fields at the level of algebraic
structures. We will show that there exists a natural L-dendriform
algebra structure on a pre-Lie algebra with a nondegenerate
symmetric 2-cocycle.

(2) {\bf ${\mathcal O}$-operators and $S$-equation in pre-Lie
algebras}. In fact, pre-Lie algebras can be regarded as the
algebraic structures behind the classical Yang-Baxter equation
(CYBE) which plays an important role in integrable systems, quantum
groups and so on (\cite{CP} and the references therein). It could be
seen more clearly in terms of $\mathcal O$-operators of a Lie
algebra introduced by Kupershmidt in \cite{K} as generalizations of
(the operator form of) the CYBE in a Lie algebra (\cite{Se}).
Explicitly, the $\mathcal O$-operators of Lie algebras provide a
direct relationship between Lie algebras and pre-Lie algebras and in
the invertible cases, they provide a necessary and sufficient
condition for the existence of a compatible pre-Lie algebra
structure on a Lie algebra. Moreover, a skew-symmetric solution of
the CYBE and the Lie algebraic analogue of the Rota-Baxter operator
(\cite{Bax,Rot}) which gives exactly the operator form of the CYBE
in \cite{Se} are understood as the $\mathcal O$-operators associated
to the co-adjoint representation and adjoint representation
respectively. Furthermore, there are some solutions of the CYBE in
certain Lie algebras obtained from pre-Lie algebras (\cite{Bai1}).

On the other hand, an analogue of the CYBE in a pre-Lie algebra,
namely, $S$-equation, was introduced in \cite{Bai2} which is closely
related to a kind of bialgebra structures on pre-Lie algebras. In
order to understand $S$-equation well, we introduce a notion of
$\mathcal O$-operator of a pre-Lie algebra in this paper and we will
show that it plays a similar role of the $\mathcal O$-operator of a
Lie algebra. Then it is natural to ask what algebraic structures
behind the $\mathcal O$-operators of pre-Lie algebras and the
related $S$-equation? The answer is L-dendriform algebras!

(3) {\bf Dendriform algebras and Loday algebras}. In fact,
L-dendriform algebras fit into a bigger framework. Recall that a
{\it dendriform algebra} $(A, \prec, \succ)$ is a vector space $A$
with two binary operations denoted by $\prec$ and $\succ $
satisfying (for any $x,y,z\in A$)
\begin{equation}(x\prec y)\prec z=x\prec (y*z),\;\;(x\succ y)\prec
z=x\succ (y\prec z),\;\;x\succ (y\succ z)=(x*y)\succ
z,\end{equation} where $x*y=x\prec y+x\succ y$. Note that $(A,*)$ is
an associative algebra as a direct consequence.

The notion of dendriform algebra was introduced by Loday (\cite{L1})
in 1995 with motivation from algebraic $K$-theory and has been
studied quite extensively with connections to several areas in
mathematics and physics, including operads, homology, Hopf algebras,
Lie and Leibniz algebras, combinatorics, arithmetic and quantum
field theory and so on (see \cite{EMP} and the references therein).
Moreover, there is the following relationship among Lie algebras,
associative algebras, pre-Lie algebras and dendriform algebras in
the sense of commutative diagram of categories (\cite{A,C,Ron}):
\begin{equation}\begin{matrix} \mbox{Lie algebra}
&\stackrel{}{\leftarrow} & \mbox{Pre-Lie algebra} \cr \uparrow
&&\uparrow  \cr
 \mbox{Associative algebra} &\stackrel{}{\leftarrow}
& \mbox{Dendriform algebra}\cr
\end{matrix}\end{equation}

Later quite a few more similar algebra structures have been
introduced, such as quadri-algebras of Aguiar and Loday (\cite{AL}).
All of them are called Loday algebras (\cite{Va}). These algebras
have a common property of ``splitting associativity", that is,
expressing the multiplication of an associative algebra as the sum
of a string of binary operations (\cite{L2}).

In order to extend the commutative diagram (1.2) at the level of
associative algebras (the bottom level of the commutative diagram
(1.2)) to the more Loday algebras, it is natural to find the
corresponding algebraic structures at the level of Lie algebras
which extends the top level of commutative diagram (1.2). We will
show that the L-dendriform algebras are chosen in a certain sense
such that the following diagram including the diagram (1.2) as a
sub-diagram is commutative:
\begin{equation}\begin{matrix} \mbox{Lie algebra} &\stackrel{}{\leftarrow} &
\mbox{Pre-Lie algebra}& \stackrel{}{\leftarrow} & \mbox{L-dendriform
algebra}\cr & \stackrel{}{\nwarrow} & &\stackrel{}{\nwarrow} &
&\stackrel{}{\nwarrow} &  \cr
 & & \mbox{Associative algebra} &\stackrel{}{\leftarrow}
& \mbox{Dendriform algebra}&\stackrel{}{\leftarrow} &
\mbox{Quadri-algebra}& \cr
\end{matrix}
\end{equation}
In this sense, L-dendriform algebras are regarded as the Lie
algebraic analogues of dendriform algebras, which explains why we
give the notion of ``L-dendriform algebra". Furthermore, it is
reasonable to consider to interpret L-dendriform algebras in terms
of the Manin black product of symmetric operads (\cite{Va}).

\subsection{Layout of the paper}
 In Section 2, we recall some basic
facts on pre-Lie algebras and introduce the notion of $\mathcal
O$-operator of a pre-Lie algebra interpreting the $S$-equation. In
Section 3, we introduce the notion of L-dendriform algebra and then
study some fundamental properties of L-dendriform algebras in terms
of the $\mathcal O$-operators of pre-Lie algebras. In particular, we
interpret their relationships with pre-Lie algebras, $S$-equation,
pseudo-Hessian structures and some Loday algebras.
  In Section 4, we introduce a notion of $\mathcal O$-operator
of an L-dendriform algebra which gives an algebraic equation
regarded as an analogue of the CYBE in a Lie algebra.
 In Section 5, we discuss certain generalization of the study
in the previous sections.

\subsection{Notations}

Throughout this paper, all algebras are finite-dimensional and over
a field of characteristic zero. We also give some notations as
follows. Let $A$ be an algebra with a binary operation $*$.

(1) Let $L_*(x)$ and $R_*(x)$ denote the left and right
multiplication operator respectively, that is, $L_*(x)y=R_*(y)x=x*y$
for any $x,y\in A$. We also simply denote them by $L(x)$ and $R(x)$
respectively without confusion. Moreover let $L_*, R_*:A\rightarrow
gl(A)$ be two linear maps with $x\rightarrow L_*(x)$ and
$x\rightarrow R_*(x)$ respectively. In particular, when $A$ is a Lie
algebra, we let ${\rm ad}(x)$ denote the adjoint operator, that is,
${\rm ad} (x)y=[x,y]$ for any $x,y\in A$.

(2) Let $r=\sum_{i}a_i\otimes b_i\in A\otimes A$. Set
\begin{equation}
r_{12}=\sum_ia_i\otimes b_i\otimes 1, r_{13}=\sum_{i}a_i\otimes
1\otimes b_i,\;r_{23}=\sum_i1\otimes a_i\otimes b_i,
\end{equation}
where $1$ is the unit if $(A,*)$ is unital or a symbol playing a
similar role of the unit for the nonunital cases. The operation
between two $r$s is in an obvious way. For example,
\begin{equation}
r_{12}*r_{13}=\sum_{i,j}a_i*a_j\otimes b_i\otimes b_j,\; r_{13}*
r_{23}=\sum_{i,j}a_i\otimes a_j\otimes
b_i*b_j,\;r_{23}*r_{12}=\sum_{i,j}a_j\otimes a_i* b_j\otimes
b_i,\end{equation} and so on. Note that Eq. (1.5) is independent of
the existence of the unit.

(3) Let $V$ be a vector space.  Let $\sigma: V\otimes V \rightarrow
V\otimes V$ be the exchanging operator defined as
\begin{equation}\sigma (x \otimes y) = y\otimes x,
\forall x, y\in V.
\end{equation}
 On
the other hand, for any $r\in A\otimes A$, define a linear map
$F_r:A^*\rightarrow A$ by
\begin{equation} \langle F_r(u^*),v^*\rangle =\langle r, u^*\otimes
v^*\rangle,\;\;\forall u^*,v^*\in A^*,
\end{equation}
where $\langle,\rangle$ is the ordinary pair between the vector
space $V$ and the dual space $V^*$. This defines an invertible
linear map $F:A\otimes A\rightarrow {\rm Hom}(A^*,A)$ and thus
allows us to identify $r$ with $F_r$ which we still denote by $r$
for simplicity of notations. Moreover, any invertible linear map
$T:V^*\rightarrow V$ can induce a nondegenerate bilinear form
$\mathcal B( , )$ on $V$ by
\begin{equation}
\mathcal B(u,v)=\langle T^{-1}u, v\rangle,\;\;\forall u,v\in V.
\end{equation}

(4) Let $V$ be a vector space. For any linear map $\rho:
A\rightarrow gl(V)$, define a linear map $\rho^* : A \rightarrow
gl(V^*)$ by
\begin{equation}
\langle \rho^*(x)v^*, u\rangle = -\langle v^*, \rho(x)u\rangle,\
\forall x \in A, u\in V, v^*\in V^*.
\end{equation}

\section{Pre-Lie algebras}
\setcounter{equation}{0}
\renewcommand{\theequation}
{2.\arabic{equation}}

\subsection{Some fundamental properties of pre-Lie algebras}

\begin{defn} {\rm Let $A$ be a vector space with a
binary operation denoted by $\circ: A\otimes A\rightarrow A$.
$(A,\circ)$ is called a {\it pre-Lie algebra} if for any $x,y,z\in
A$, the associator
\begin{equation}
(x,y,z)=(x\circ y)\circ z-x\circ (y\circ z) \end{equation} is
symmetric in $x,y$, that is,
\begin{equation}
(x,y,z)=(y,x,z),\;{\rm or}\;{\rm equivalently}\;(x\circ y)\circ
z-x\circ (y\circ z)=(y\circ x)\circ z-y\circ (x\circ z),\;\forall
x,y,z\in A.
\end{equation}
}
\end{defn}

\begin{prop} {\rm (cf. \cite{G,Vi})}\quad Let $(A, \circ)$ be a pre-Lie algebra.

{\rm (1)} The commutator
\begin{equation} [x,y] =
x\circ y - y\circ x,\;\;\forall x,y\in A
\end{equation} defines a Lie algebra $\frak g(A)$, which is called the sub-adjacent Lie algebra of
$A$ and $A$ is also called a compatible pre-Lie algebra structure
on the Lie algebra $\frak g (A)$.

{\rm (2)} $L_\circ$ gives a representation of the Lie algebra $\frak
g(A)$, that is,
\begin{equation}
L_\circ({[x, y]}) = L_\circ(x)L_\circ(y)- L_\circ(y)L_\circ(x),\
\forall x, y\in A.
\end{equation}
\end{prop}

\begin{prop} Let $\frak g $ be a vector space with a binary operation $\circ$.
Then $(\frak g , \circ)$ is a pre-Lie algebra if and only if $(\frak
g , [,])$ defined by Eq. (2.3) is a Lie algebra and $(L_\circ, \frak
g )$ is a representation.
\end{prop}

\begin{defn}{\rm (\cite{Bai2})\quad Let $(A,\circ)$ be a pre-Lie algebra and $V$ be a vector space. Let
$l,r: A\rightarrow gl(V)$ be two linear maps. $(l,r,V)$ is called a
{\it module of $(A,\circ)$} if
\begin{equation}
l(x)l(y)-l(x\circ y)=l(y)l(x)-l(y\circ x),
\end{equation}
\begin{equation}
l(x)r(y)-r(y)l(x)=r(x\circ y)-r(y)r(x),\forall x,y \in A.
\end{equation}}\end{defn}

In fact, $(l,r, V)$ is a module of a pre-Lie algebra $(A,\circ)$ if
and only if the direct sum $A\oplus V$ of the underlying vector
spaces of $A$ and $V$ is turned into a pre-Lie algebra (the {\it
semidirect sum}) by defining multiplication in $A\oplus V$ by
\begin{equation}
(x+u)*(y+v)=x\circ y+(l(x)v+r(y)u),\;\;\forall x,y\in A, u,v\in
V.\end{equation} We denote it by $A\ltimes_{l,r}V$.

\begin{prop} {\rm (\cite{Bai2})}\quad Let $(l,r,V)$ be a module
of a pre-Lie algebra $(A,\circ)$.  Then $(l^*-r^*, -r^*, V^*)$ is a
module of $(A,\circ)$. \end{prop}

\begin{defn}
{\rm Let $(A,\circ)$ be a pre-Lie algebra and $\mathcal {B}$ be a
bilinear form on $A$. $\mathcal{B}$ is called a {\it 2-cocycle of
$A$} if $\mathcal {B}$ satisfies
\begin{equation}
\mathcal {B}(x\circ y, z) - \mathcal {B}(x, y\circ z) = \mathcal
{B}(y\circ x, z) - \mathcal {B}(y, x\circ z),\forall x, y, z\in A.
\end{equation}}\end{defn}

\begin{remark}{\rm
A real pre-Lie algebra $A$ endowed with a symmetric and definite
positive 2-cocycle corresponds to an affine Lie group $G$ with a
$G$-invariant Hessian metric (\cite{Sh}). }\end{remark}

\begin{defn}{\rm
Let $(A, \circ)$ be a pre-Lie algebra and $r \in A\otimes A$. The
following equation is called {\it $S$-equation in $(A,\circ)$}:
\begin{equation}
-r_{12}\circ r_{13} + r_{12}\circ r_{23} + [r_{13}, r_{23}] = 0.
\end{equation}}
\end{defn}

The $S$-equation in a pre-Lie algebra is an analogue of the CYBE in
a Lie algebra, which is related to the study of pre-Lie bialgebras
(\cite{Bai2}).

\subsection{$\mathcal {O}$-operators of pre-Lie algebras and $S$-equation}

\begin{defn}{\rm Let $(A, \circ)$ be a
pre-Lie algebra and $(l, r, V)$ be a module. A linear map $T:
V\rightarrow A$ is called an {\it $\mathcal {O}$-operator associated
to $(l, r, V)$}  if $T$ satisfies
\begin{equation}
T(u)\circ T(v) = T(l(T(u))v + r(T(v))u), \forall u, v \in V.
\end{equation}}\end{defn}

\begin{exam}{\rm Let $(A,\circ)$ be a pre-Lie algebra.
A linear map $R: A \rightarrow A$ is called a {\it Rota-Baxter
operator (of weight 0)} on $A$ if $R$ is an $\mathcal {O}$-operator
associated to the module $(L, R, A)$, that is, $R$ satisfies
(\cite{LHB})
\begin{equation}
R(x) \circ R(y) = R(R(x) \circ y + x \circ R(y)), \forall x, y \in
A.
\end{equation}
}\end{exam}

\begin{remark} {\rm In the case of associative algebras, the linear
map $T$ satisfying Eq. (2.10) was introduced independently  in
\cite{U} under a notion of {\it generalized Rota-Baxter operator}.
}\end{remark}

\begin{theorem}
Let $(A, \circ)$ be a pre-Lie algebra and $r\in A\otimes A$ be
symmetric. Then  $r$ is a solution of $S$-equation in $A$ if and
only if $r$ is an $\mathcal {O}$-operator of $(A,\circ)$ associated
to $(L_\circ^*-R_\circ^*, -R_\circ^*, A^*)$.
\end{theorem}

\begin{proof} Let $\{e_1, ..., e_n\}$ be a basis of $A$ and $\{e_1^*, ...,
e_n^*\}$ be the dual basis. Suppose that $e_i \circ e_j =
\sum\limits_{k=1}^n c_{ij}^k e_k$ and $r = \sum\limits_{i,j=1}^n
a_{ij} e_i\otimes e_j$, $a_{ij} = a_{ji}$. Hence $r(e^*_i) =
\sum\limits_{k=1}^n a_{ik}e_k$. Then the coefficient of $e_k$ in
$$r(e^*_i)\circ r(e_j^*) - r((L_\circ^*-R_\circ^*)(r(e^*_i))e_j^* -
R^*_\circ(r(e_j^*))e_i^*) =0,$$ is
$$\sum_{t,l=1}^n (a_{it} a_{jl} c_{tl}^k +a_{it} a_{lk} c_{tl}^j -
a_{it} a_{lk} c_{lt}^j - a_{jt} a_{lk} c_{lt}^i)=0,$$ which is
precisely the coefficient of $e_i\otimes e_j\otimes e_k$ in the
equation
$$r_{13}\circ r_{23} + [r_{12}, r_{23}] - r_{13}\circ r_{12} = 0,$$
which is an equivalent form of $S$-equation (\cite{Bai2}). Therefore
the conclusion holds.
\end{proof}

\begin{theorem}
Let $(A,\circ)$ be a pre-Lie algebra and $(l, r, V)$ be a module.
Let $T : V \rightarrow A$ be a linear map which is identified as an
element in the vector space $(A\oplus V^*)\otimes (A\oplus V^*)$.
Then $r = T + \sigma(T)$ is a symmetric solution of $S$-equation in
the pre-Lie algebra $A\ltimes_{l^*-r^*, -r^*} V^*$ if and only if
$T$ is an $\mathcal {O}$-operator of $(A,\circ)$ associated to $(l,
r, V )$.
\end{theorem}

\begin{proof}  Let
$\{v_1,...,v_m\}$ be a basis of $V$ and $\{v^*_1,...,v^*_m\}$ be the
dual basis. Then
$$T = \sum_{i=1}^m T(v_i) \otimes v_i^* \in T(V) \otimes V^* \subset
(A\oplus V^*) \otimes (A\oplus V^*).$$ By Eq. (1.9), we show that
$$l^*(T(v_i))v_j = -\sum_{k=1}^m
v_j(l(T(v_i))v_k)v_k^*,r^*(T(v_i))v_j = -\sum_{k=1}^m
v_j(r(T(v_i))v_k)v_k^*.$$ Therefore we get
{\small \begin{eqnarray*} &&-r_{12} \circ r_{13} \\
&=& \sum_{i,j=1}^m [-(T(v_i)\circ T(v_j)\otimes v_i^* \otimes v_j^*
+ (l^* - r^*)(T(v_i))v_j^*\otimes v_i^*\otimes
T(v_j) -r^*(T(v_j))v_i^*\otimes T(v_i)\otimes v_j^*)]\\
%&=& -\sum_{i,j=1}^m T(v_i)\circ T(v_j)\otimes v_i^* \otimes
%v_j^*+\sum_{i,j=1}^m\sum_{k=1}^m v_j^*((l-r)(T(v_i))v_k)v_k^*\otimes
%v_i^*\otimes T(v_j)\\
%&&\hspace{0cm}- \sum_{i,j=1}^m\sum_{k=1}^m
%v_i^*(r(T(v_j)v_k))v_k^* \otimes T(v_i)\otimes v_j^* \\
&=& \sum_{i,j=1}^m [-T(v_i)\circ T(v_j)\otimes v_i^* \otimes v_j^* +
v_j^* \otimes v_i^*\otimes T((l-r)(T(v_i))v_j)- v_i^*\otimes
T(r(T(v_j))v_i)\otimes v_j^*].
\end{eqnarray*}}
Similarly, we have {\small \begin{eqnarray*} r_{12}\circ r_{23} &=&
\sum_{i,j=1}^m[T(r(T(v_j))v_i)\otimes v_i^*\otimes v_j^* -
 v_i^* \otimes v_j^*\otimes T((l-r)(T(v_i))v_j)+ v_i^*\otimes T(v_i)\circ T(v_j) \otimes v_j^*],
\end{eqnarray*}}
$$[r_{13}, r_{23}] = \sum_{i,j=1}^m (T(l(T(v_j))v_i)\otimes
v_j^*\otimes v_i^* - v_i^*\otimes T(l(T(v_i))v_j) \otimes v_j^* +
v_i^*\otimes v_j^*\otimes [T(v_i), T(v_j)]).$$
%\end{eqnarray*}
So $r$ is a symmetric solution of $S$-equation in the pre-Lie
algebra $A\ltimes_{l^*-r^*, -r^*} V^*$ if and only if $T$ is an
$\mathcal {O}$-operator of $(A, \circ)$ associated to $(l, r, V )$.
\end{proof}

Combining Theorem 2.12 and Theorem 2.13, we have the following
conclusion:

\begin{coro} Let $(A,\circ)$ be a pre-Lie algebra and $(l, r, V)$ be a module.
  Set $\hat
A=A\ltimes_{l^*-r^*,-r^*}V^*$. Let $T:V\rightarrow A$ be a linear
map. Then the following conditions are equivalent:

{\rm (1)}  $T$ is an $\mathcal {O}$-operator of $(A,\circ)$
associated to $(l, r, V )$.

{\rm (2)} $T+\sigma(T)$ is a symmetric solution of $S$-equation in
the pre-Lie algebra $\hat A$.

{\rm (3)} $T+\sigma(T)$ is an $\mathcal {O}$-operator of $\hat A$
associated to $(L_{\hat A}^*-R_{\hat A}^*, -R_{\hat A}^*, {\hat
A}^*)$.
\end{coro}

\section{L-dendriform algebras}
\setcounter{equation}{0}
\renewcommand{\theequation}
{3.\arabic{equation}}

\subsection{Definition and some basic properties}

\begin{defn}{\rm Let $A$ be a vector space with two
binary operations denoted by $\triangleright$ and $ \triangleleft:
A\otimes A \rightarrow A$.  $(A, \triangleright, \triangleleft)$ is
called an {\it L-dendriform algebra} if for any $x,y,z\in A$,
\begin{equation}
x\triangleright(y\triangleright z)=(x\triangleright y)\triangleright
z + (x\triangleleft y)\triangleright z +
y\triangleright(x\triangleright z) - (y\triangleleft
x)\triangleright z - (y\triangleright x)\triangleright z,
\end{equation}
\begin{equation}
x\triangleright(y\triangleleft z)=(x\triangleright y)\triangleleft z
+ y\triangleleft (x\triangleright z) + y\triangleleft(x\triangleleft
z) - (y\triangleleft x)\triangleleft z.
\end{equation}}\end{defn}

\begin{prop} Let $(A,\triangleright,
\triangleleft)$ be an L-dendriform algebra.

{\rm (1)} The binary operation $\bullet:A\otimes A\rightarrow A$
given by
\begin{equation}
x\bullet y=x\triangleright y + x \triangleleft y,\forall x, y \in A,
\end{equation}
defines a pre-Lie algebra. $(A,\bullet)$ is called the associated
horizontal pre-Lie algebra of $(A,\triangleright, \triangleleft)$
and $(A,\triangleright, \triangleleft)$ is called a compatible
L-dendriform algebra structure on the pre-Lie algebra $(A,\bullet)$.

{\rm (2)} The binary operation $\circ:A\otimes A\rightarrow A$ given
by
\begin{equation}
x\circ y=x\triangleright y - y \triangleleft x,\forall x, y \in A,
\end{equation}
defines a pre-Lie algebra. $(A,\circ)$ is called the associated
vertical pre-Lie algebra of $(A,\triangleright, \triangleleft)$ and
$(A,\triangleright, \triangleleft)$ is called a compatible
L-dendriform algebra structure on the pre-Lie algebra $(A,\circ)$.

{\rm (3)} Both $(A,\bullet)$ and $(A,\circ)$ have the same
sub-adjacent Lie algebra $\frak g(A)$ defined by
\begin{equation}
[x,y]=x\triangleright y + x \triangleleft y - y\triangleright x - y
\triangleleft x, \forall x,y\in A.
\end{equation}
\end{prop}

\begin{proof}
It is straightforward. \end{proof}

\begin{remark}{\rm Let $(A,\triangleright,
\triangleleft)$ be an L-dendriform algebra. Then Eqs. (3.1) and
(3.2) can be rewritten as (for any $x,y,z\in A$)
\begin{equation}
x\triangleright(y\triangleright z)-(x\bullet y)\triangleright z =
y\triangleright(x\triangleright z) - (y\bullet x)\triangleright z,
\end{equation}
\begin{equation}
x\triangleright(y\triangleleft z) -(x\triangleright y)\triangleleft
z= y\triangleleft (x\bullet z)-(y\triangleleft x)\triangleleft z.
\end{equation}
The both sides of the above two equations can be regarded as a kind
of ``generalized associators". In this sense, Eqs. (3.6) and (3.7)
express certain ``generalized left-symmetry" of the ``generalized
associators".}\end{remark}

\begin{prop} Let $A$ be a vector space with two
binary operations denoted by $\triangleright, \triangleleft:
A\otimes A \rightarrow A$.

{\rm (1)} $(A, \triangleright, \triangleleft)$ is an L-dendriform
algebra if and only if $(A,\bullet)$ defined by Eq. (3.3) is a
pre-Lie algebra and $(L_\triangleright, R_\triangleleft, A)$ is a
module.

{\rm (2)} $(A, \triangleright, \triangleleft)$ is an L-dendriform
algebra if and only if $(A,\circ)$ defined by Eq. (3.4) is a pre-Lie
algebra and $(L_\triangleright, -L_\triangleleft, A)$ is a module.
\end{prop}
\begin{proof}
The conclusions can be obtained by a straightforward computation or
a similar proof as of Theorem 3.8.
\end{proof}

\begin{coro}
Let $(A,\triangleright, \triangleleft)$ be an L-dendriform algebra.
Then $(L_\triangleright^* - R_\triangleleft^*,-
R_\triangleleft^*,A^*)$ is a module of the associated horizontal
pre-Lie algebra  $(A, \bullet)$ and $(L_\triangleright^* +
L_\triangleleft^*,\ L_\triangleleft^*,A^*)$ is a module of the
associated vertical pre-Lie algebra $(A, \circ)$.
\end{coro}

\begin{proof}
It follows from Proposition 2.5 and Proposition 3.4.
\end{proof}

\begin{prop} Let $(A,\triangleright, \triangleleft)$ be an L-dendriform
algebra. Define two binary operations $\triangleright^t,
\triangleleft^t: A\otimes A\rightarrow A$ by
\begin{equation}
x\triangleright^t y=x\triangleright y,\;\; x\triangleleft^t
y=-y\triangleleft x,\;\;\forall x,y\in A.
\end{equation}
Then $(A,\triangleright^t, \triangleleft^t)$ is an L-dendriform
algebra. The associated horizontal pre-Lie algebra of
$(A,\triangleright^t$, $\triangleleft^t)$ is the associated vertical
pre-Lie algebra $(A,\circ)$ of $(A,\triangleright, \triangleleft)$
and the associated vertical pre-Lie algebra of $(A,\triangleright^t,
\triangleleft^t)$ is the associated horizontal pre-Lie algebra
$(A,\bullet)$ of $(A,\triangleright, \triangleleft)$, that is,
\begin{equation}
\bullet^t=\circ,\;\;\;\circ^t=\bullet.
\end{equation}\end{prop}

\begin{proof}
It is straightforward.\end{proof}

\begin{defn} {\rm Let $(A,\triangleright, \triangleleft)$ be an L-dendriform
algebra. The L-dendriform algebra $(A,\triangleright^t,
\triangleleft^t)$ given by Eq. (3.8) is called  the {\it transpose}
of $(A,\triangleright, \triangleleft)$.}
\end{defn}

For brevity, in the following (sub)sections, we only give the study
involving the associated vertical pre-Lie algebras. The
corresponding study on the associated horizontal pre-Lie algebras
can be obtained by the transposes of the L-dendriform algebras
through Proposition 3.6.

\subsection{L-dendriform algebras and $\mathcal O$-operators of
pre-Lie algebras}

\begin{theorem} Let $(A, \circ)$ be a pre-Lie algebra and $(l, r, V)$ be a
module. If $T$ is an $\mathcal {O}$-operator associated to $(l, r,
V)$, then there exists an L-dendriform algebra structure on $V$
defined by
\begin{equation}
u \triangleright v = l(T(u))v,u \triangleleft v = -r(T(u))v, \forall
u, v\in V.
\end{equation}
Therefore there is a pre-Lie algebra structure on $V$ defined by Eq.
(3.4) as the associated vertical pre-Lie algebra of $(V,
\triangleright, \triangleleft)$ and $T$ is a homomorphism of pre-Lie
algebras. Furthermore, $T(V) = \{T(v)\mid v\in V\}\subset A$ is a
pre-Lie subalgebra of $(A,\circ)$ and there is an induced
L-dendriform algebra structure on $T(V)$ given by
\begin{equation}
T(u) \triangleright T(v) = T(u \triangleright v),T(u) \triangleleft
T(v) = T(u \triangleleft v), \forall u, v \in V.
\end{equation}
Moreover, the corresponding associated vertical pre-Lie algebra
structure on $T(V)$  is a pre-Lie subalgebra of $(A, \circ)$ and $T$
is a homomorphism of L-dendriform algebras.
\end{theorem}

\begin{proof}
For any $u, v, w \in V$, we have {\small \begin{eqnarray*} u
\triangleright (v \triangleright w) &=& l(T(u))l(T(v))w,(u
\triangleright v) \triangleright w = l(T(l(T(u))v))w, (u
\triangleleft v) \triangleright w = -l(T(r(T(u))v))w,\\
v\triangleright (u \triangleright w) &=& l(T(v))l(T(u))w, (v
\triangleleft u) \triangleright w = -l(T(r(T(v))u))w, (v
\triangleright u) \triangleright w = l(T(l(T(v))u))w.\\
(u \triangleright v) \triangleleft w&=& -r(T(l(T(u))v))w, u
\triangleright (v \triangleleft w) = -l(T(u))r(T(v))w,  v
\triangleleft (u \triangleright w) = -r(T(v))l(T(u))w,\\
v \triangleleft (u \triangleleft w) &=& r(T(v))r(T(u))w,(v
\triangleleft u) \triangleleft w = r(T(r(T(v))u))w.
\end{eqnarray*}}
Hence
\begin{eqnarray*}
&&(u \triangleright v) \triangleright w + (u \triangleleft v)
\triangleright w + v\triangleright (u \triangleright w) - (v
\triangleleft u) \triangleright w - (v \triangleright u)
\triangleright w - u \triangleright (v \triangleright w)\\
&&\hspace{0.5cm}= l(T(u))l(T(v))w-l(T(v))l(T(u))w-l(T(u)\circ
T(v))T(v)w +l(T(v)\circ T(u))w=0,\\
&& (u \triangleright v) \triangleleft w + v\triangleleft (u
\triangleright w) + v \triangleleft (u \triangleleft w) - (v
\triangleleft u) \triangleleft w - u \triangleright (v \triangleleft
w)\\
&&\hspace{0.5cm} = -r(T(l(T(u))v))w + r(T(u)\circ T(v))w -
r(T(r(T(v))u))w=0.
\end{eqnarray*}
Therefore $(V, \triangleright, \triangleleft)$ is an L-dendriform
algebra. The other conclusions follow easily.
\end{proof}

\begin{coro} Let $(A,\circ)$ be a pre-Lie algebra and $R$ be a Rota-Baxter
operator of weight zero. Then the binary operations given by
\begin{equation}
x \triangleright y = R(x)\circ y,x \triangleleft y = -y\circ R(x),
\forall x, y \in A,
\end{equation}
define an L-dendriform algebra structure on $A$.
\end{coro}

\begin{proof}
It follows immediately from Theorem 3.8 by taking $V=A$, $l=L$ and
$r=R$.
\end{proof}

Recall that a linear map $T:V \rightarrow \frak g $ is called an
{\it $\mathcal {O}$-operator of a Lie algebra $\frak g$ associated
to a representation $(\rho, V)$} if $T$ satisfies
\begin{equation}
[T(u), T(v)] = T(\rho(T(u))v - \rho(T(v))u),\forall u,v \in V.
\end{equation}
In particular, if $R$ is an $\mathcal {O}$-operator of $\frak g $
associated to the representation $({\rm {\rm ad} }, \frak g )$, it
is known (\cite{GS}) that there exists a pre-Lie algebra structure
on $\frak g$ given by
\begin{equation}
x\circ y = [R(x), y], \forall x ,y \in \frak g.
\end{equation}

\begin{coro} Let $\frak g $ be a Lie algebra and
 $\{R_1,R_2\}$ be a pair of commuting $\mathcal {O}$-operators
of $\frak g $ associated to  $({\rm ad}, \frak g )$. Then there
exists an L-dendriform algebra structure on $\frak g $ defined by
\begin{equation}
x\triangleright y = [R_1(R_2(x)), y],x\triangleleft y = [R_2(x),
R_1(y)],\forall x, y \in \frak g .
\end{equation}
\end{coro}

\begin{proof} There exists a pre-Lie algebra
structure on $\frak g$ defined by Eq. (3.14) with the $\mathcal
O$-operator $R_1$ of the Lie algebra $\frak g$ associated to $({\rm
ad} , \frak g )$. It is straightforward to show that $R_2$ is a
Rota-Baxter operator of weight zero on this pre-Lie algebra if $R_2$
as an $\mathcal O$-operator of the Lie algebra $\frak g$ associated
to $({\rm ad}, \frak g )$ is commutative with $R_1$. Then the result
follows from Corollary 3.9.
\end{proof}

\begin{theorem} Let $(A,\circ)$ be a pre-Lie algebra. Then
there exists a compatible L-dendriform algebra structure on
$(A,\circ)$ such that $(A,\circ)$ is the associated vertical pre-Lie
algebra if and only if there exists an invertible $\mathcal
{O}$-operator of $(A,\circ)$.
\end{theorem}

\begin{proof}
Suppose that there exists an invertible ${\mathcal O}$-operator of
$(A,\circ)$ associated to a module $(l,r,V)$. By Theorem 3.8, there
exists an L-dendriform algebra structure on $V$ given by Eq. (3.10).
Therefore we define an L-dendriform algebra structure on $A$ by Eq.
(3.11) such that $T$ is an isomorphism of L-dendriform algebras,
that is,
$$x\triangleright y = T(l(x)T^{-1}(y)),x\triangleleft y = -T(r(x)T^{-1}(y)),\;\;\forall x,y\in A.$$
Moreover it is a compatible L-dendriform algebra structure on
$(A,\circ)$ since
$$x\triangleright y - y\triangleleft x = T(l(x)T^{-1}(y) +
r(y)T^{-1}(x)) = T(T^{-1}(x)\circ T^{-1}(y))=x\circ y,\;\;\forall
x,y\in A.$$ Conversely, let $(A,\triangleright,\triangleleft)$ be an
L-dendriform algebra and $(A,\circ)$ be the associated vertical
pre-Lie algebra. Then $(L_\triangleright,-L_\triangleleft,A)$ is a
module of $(A,\circ)$ and the identity map $id: A\rightarrow A$ is
an $\mathcal {O}$-operator of $(A, \circ)$ associated to
$(L_\triangleright,-L_\triangleleft,A)$. \end{proof}

The following conclusion reveals the relationship between
L-dendriform algebras and pseudo-Hessian structures (that is, the
pre-Lie algebras a nondegenerate symmetric 2-cocycle):

\begin{coro} Let $(A, \circ)$ be a pre-Lie algebra
with a nondegenerate symmetric 2-cocycle $\mathcal {B}$. Then there
exists a compatible L-dendriform algebra structure on $(A,\circ)$
given by
\begin{equation}
\mathcal {B}(x\triangleright y,z)= -\mathcal {B}(y,[x,z]),\;\;
\mathcal {B}(x\triangleleft y,z)= -\mathcal {B}(y,z\circ x), \forall
x, y, z \in A.
\end{equation}
such that $(A,\circ)$ is the associated vertical pre-Lie algebra.
\end{coro}

\begin{proof}
It is straightforward to show that the invertible linear map
$T:A^*\rightarrow A$ defined by Eq. (1.8) is an invertible $\mathcal
O$-operator of $(A,\circ)$ associated to the module $(L_\circ^* -
R_\circ^*,- R_\circ^*,A^*)$. By Theorem 3.11, there is a compatible
L-dendriform algebra structure on $A$ defined by (for any $x,y,z\in
A$)
\begin{eqnarray*} \mathcal {B}(x\triangleright y, z) &=&
\mathcal {B}(T((L_\circ^* - R_\circ^*)(x)T^{-1}(y)), z) = \langle
(L_\circ^* - R_\circ^*)(x)T^{-1}(y), z\rangle = - \langle T^{-1}(y),
[x, z]
\rangle\\&  =& - \mathcal {B}(y, [x,z]);\\
\mathcal {B}(x\triangleleft y,z)&=&\mathcal
B(T(R_\circ^*(x)T^{-1}(y)),z)= \langle
R_\circ^*(x)T^{-1}(y),z\rangle=-\langle T^{-1}(y), z\circ x\rangle=
 -\mathcal {B}(y,z\circ x).
\end{eqnarray*} such that $(A,\circ)$ is the associated vertical pre-Lie
algebra. Hence the conclusion holds.
\end{proof}

The following conclusion provides a construction of solutions of
$S$-equation in certain pre-Lie algebras from L-dendriform algebras:

\begin{coro}
Let $(A, \triangleright, \triangleleft)$ be an L-dendriform algebra
and $(A, \circ),(A, \bullet)$ be the associated vertical and
horizontal pre-Lie algebras respectively. Let $\{e_1,..., e_n\}$ be
a basis of $A$ and $\{e_1^*,..., e_n^*\}$ be the dual basis. Then $
r = \sum\limits_{i=1}^n (e_i\otimes e_i^* + e_i^*\otimes e_i)$ is a
symmetric solution of $S$-equation in  the pre-Lie algebras
$A\ltimes_{L_\triangleright^* + L_\triangleleft^*,
L_\triangleleft^*}A^*$ and $A\ltimes_{L_\triangleright^* -
R_\triangleleft^*, -R_\triangleleft^*}A^*$.
\end{coro}

\begin{proof} Since  $id$ is an $\mathcal {O}$-operator of both the
pre-Lie algebra $(A, \circ)$ associated to the module
$(L_\triangleright,-L_\triangleleft,A)$ and the pre-Lie algebra $(A,
\bullet)$ associated to the module $(L_\triangleright,\
R_\triangleleft,A)$, the conclusion follows from Theorem 2.13.
\end{proof}

\subsection{Relationships with dendriform algebras and quadri-algebras}

\begin{prop}
 Any dendriform algebra
$(A,\succ,\prec)$ is an L-dendriform algebra by letting
$x\triangleright y=x\succ y, x\triangleleft y=x\prec y$.
\end{prop}

\begin{proof}
In fact, for a dendriform algebra, both sides of Eqs. (3.6) and
(3.7) which are the equivalent identities of an L-dendriform
algebra, are zero.
\end{proof}

\begin{remark}
{\rm In the above sense, associative algebras are the special
pre-Lie algebras whose associators are zero, whereas dendriform
algebras are the special L-dendriform algebras whose ``generalized
associators" (see Remark (3.3)) are zero. }
\end{remark}

By Proposition 3.2 and Proposition 3.14, the following result is obvious:

\begin{coro}
Let $(A,\succ, \prec)$ be a dendriform algebra.

{\rm (1)} {\rm (\cite{L1})} The binary operation given by Eq. (3.3)
defines a pre-Lie algebra (in fact, it is an associative algebra).

{\rm (2)} {\rm (\cite{C,Ron})} The binary operation given by Eq.
(3.4) defines a pre-Lie algebra.

\end{coro}

\begin{defn}{\rm (\cite{AL})\quad Let $A$ be a vector space
 with four binary operations denoted by
$\searrow, \nearrow, \nwarrow$ and $\swarrow : A\otimes A\rightarrow
A$. $(A, \searrow, \nearrow, \nwarrow, \swarrow)$ is called a {\it
quadri-algebra} if for any $x,y,z\in A$,
\begin{equation}
(x\nwarrow y)\nwarrow z=x\nwarrow (y* z),\;\; (x\nearrow y)\nwarrow
z=x\nearrow (y\prec z),\;\;(x\wedge y)\nearrow z=x\nearrow (y\succ
z),\end{equation}
\begin{equation}
(x\swarrow y)\nwarrow z=x\swarrow (y\wedge z),\;\; (x\searrow
y)\nwarrow z=x\searrow (y\nwarrow z),\;\;(x\vee y)\nearrow
z=x\searrow (y\nearrow z),\end{equation}
\begin{equation}(x\prec y)\swarrow z=x\swarrow (y\vee z),\;\; (x\succ y)\swarrow
z=x\searrow (y\swarrow z),\;\;(x* y)\searrow z=x\searrow (y\searrow
z),\end{equation}
 where
\begin{equation}x\succ y=x\nearrow y+x\searrow y,
x\prec y=x\nwarrow y+x\swarrow y,\end{equation}\begin{equation}
x\vee y=x\searrow y+x\swarrow y,x\wedge y=x\nearrow y+x\nwarrow
y,\end{equation}
\begin{equation}x*y=x\searrow y+x\nearrow y+x\nwarrow y+x\swarrow y=x\succ
y+x\prec  y=x\vee y+x\wedge y.\end{equation} }\end{defn}

\begin{prop}{\rm (\cite{AL})}\quad Let $(A, \searrow, \nearrow, \nwarrow,
\swarrow)$ be a quadri-algebra.

{\rm (1)} The binary operations given by Eq. (3.20) define a
dendriform algebra $(A,\succ,\prec)$.

{\rm (2)} The binary operations given by Eq. (3.21)
 define a dendriform algebra $(A,\vee,\wedge)$.

{\rm (3)} The binary operation given by  Eq. (3.22) defines an
associative algebra $(A,*)$.
\end{prop}

There is an additional relation between quadri-algebras and
L-dendriform algebras as follows:

\begin{prop}
Let $(A, \searrow, \nearrow, \nwarrow, \swarrow)$ be a
quadri-algebra. The binary operations
 given by \begin{equation} x
\triangleright y = x \searrow y - y \nwarrow x,x \triangleleft y = x
\nearrow y - y \swarrow x,\;\;\forall x,y\in A,
\end{equation} define an L-dendriform algebra $(A, \triangleright,\triangleleft)$.
\end{prop}

\begin{proof}
It is straightforward.
\end{proof}

\begin{coro} Let $(A, \searrow, \nearrow, \nwarrow, \swarrow)$ be a
quadri-algebra.

{\rm (1)} The binary operation given by
\begin{equation} x \circ y = x \searrow y + x \swarrow y
- y \nwarrow x - y \nearrow x = x \triangleright y -y \triangleleft
x =x \vee y - y \wedge x,\;\;\forall x,y\in A,
\end{equation}
defines a pre-Lie algebra $(A,\circ)$.

{\rm (2)} The binary operation given by Eq. (3.22) defines an
associative algebra $(A,*)$.

 {\rm (3)} The binary operation given by
\begin{equation} x \bullet y = x \searrow y + x \nearrow y
- y \nwarrow x - y \swarrow x= x \triangleright y +x \triangleleft y
=x \succ y -y \prec x,\;\;\forall x,y\in A,
\end{equation}
defines a pre-Lie algebra $(A,\bullet)$.

{\rm (4)} The binary operation given by
\begin{equation}
[x, y] = x \searrow y + x \swarrow y + x \nearrow y + x \nwarrow y -
(y \searrow x + y \swarrow x + y \nearrow x + y \nwarrow
x),\;\;\forall x,y\in A,
\end{equation}
defines a Lie algebra $(\frak g(A),[,])$.

\end{coro}

\begin{proof}

(1) and (3) follow from Proposition 3.18, Proposition 3.19 and Corollary
3.16. (2) is exactly the conclusion (3) in Proposition 3.18 which
follows from the conclusions (1) and (2) in Proposition 3.18 and
Corollary 3.16. (4) follows from (1), (2), (3) and Proposition
2.2.
\end{proof}

Summarizing the above study in this subsection, we have the
following commutative diagram:
\begin{equation}\begin{matrix} \mbox{Lie algeba} &\stackrel{-}{\leftarrow} &
\mbox{Pre-Lie algebra}& \stackrel{-,+}{\leftarrow} &
\mbox{L-dendriform algebra}& \cr & \stackrel{-}{\nwarrow} &
\Uparrow\in &\stackrel{-}{\nwarrow} & \Uparrow \in
&\stackrel{-}{\nwarrow} \cr
 & & \mbox{Associative algebra} &\stackrel{+}{\leftarrow}
& \mbox{Dendriform algebra}&\stackrel{+}{\leftarrow} &
\mbox{Quadri-algebra}\cr
\end{matrix}\end{equation}
\noindent where $``\Uparrow\in ''$ means the inclusion. $``+''$
means the binary operation $x\circ_1y+x\circ_2 y$ and  $``-''$ means
the binary operation $x\circ_1y-y\circ_2 x$.

\section{$\mathcal {O}$-operators of L-dendriform algebras and $LD$-equation}
 \setcounter{equation}{0}
\renewcommand{\theequation}
{4.\arabic{equation}}

%\subsection{Bimodules of L-dendriform algebras}

\begin{defn}{\rm Let $(A,\triangleright, \triangleleft)$ be an L-dendriform
algebra and $V$ be a vector space. Let $l_{\triangleright}$,
$r_{\triangleright},l_{\triangleleft},r_{\triangleleft}:A\rightarrow
gl(V)$ be four linear maps.
$(l_{\triangleright},r_{\triangleright},l_{\triangleleft},r_{\triangleleft},
V)$ is called a {\it module of $(A,\triangleright, \triangleleft)$}
if
%the following five equations hold  (for any $x, y \in A$)
\begin{equation}
[l_{\triangleright}(x),l_{\triangleright}(y)] =
l_{\triangleright}[x,y];
\end{equation}
\begin{equation}
[l_{\triangleright}(x),l_{\triangleleft}(y)] =
l_{\triangleleft}(x\circ y) +
l_{\triangleleft}(y)l_{\triangleleft}(x);
\end{equation}
\begin{equation}
r_{\triangleright}(x\triangleright y) =
r_{\triangleright}(y)r_{\triangleright}(x) +
r_{\triangleright}(y)r_{\triangleleft}(x) +
[l_{\triangleright}(x),r_{\triangleright}(y)]-
r_{\triangleright}(y)l_{\triangleleft}(x);
\end{equation}
\begin{equation}
r_{\triangleright}(x\triangleleft y) =
r_{\triangleleft}(y)r_{\triangleright}(x) +
l_{\triangleleft}(x)r_{\triangleright}(y) +
[l_{\triangleleft}(x),r_{\triangleleft}(y)];
\end{equation}
\begin{equation}
[l_{\triangleright}(x),r_{\triangleleft}(y)] =
r_{\triangleleft}(x\bullet y) -
r_{\triangleleft}(y)r_{\triangleleft}(x),\forall x,y\in A,
\end{equation}
where $x\circ y= x\triangleright y - y\triangleleft x,$ and $
x\bullet y= x\triangleright y + x\triangleleft y$.}\end{defn}

In fact,
$(l_{\triangleright},r_{\triangleright},l_{\triangleleft},r_{\triangleleft},
V)$ is a module of an L-dendriform algebra $(A, \triangleright,
\triangleleft)$ if and only if there exists an L-dendriform algebra
structure on the direct sum $A\oplus V$ of the underlying vector
spaces of $A$ and $V$ given by (for any $x,y \in A, u,v \in V $)
\begin{equation} (x+u)\triangleright(y+v)=x\triangleright y +
l_{\triangleright}(x)v + r_{\triangleright}(y)u,
(x+u)\triangleleft(y+v)=x\triangleleft y + l_{\triangleleft}(x)v +
r_{\triangleleft}(y)u.
\end{equation}
We denote it by
$A\ltimes_{l_{\triangleright},r_{\triangleright},l_{\triangleleft},r_{\triangleleft}}V.$

\begin{prop}
Let $(l_{\triangleright}, r_{\triangleright}, l_{\triangleleft},
r_{\triangleleft}, V)$ be a module of an L-dendriform algebra
$(A,\triangleright, \triangleleft)$. Then
$(l_{\triangleright}^*+l_{\triangleleft}^*-r_{\triangleright}^*-r_{\triangleleft}^*,\
r_{\triangleright}^*,r_{\triangleright}^*-l_{\triangleleft}^*,\
-(r_{\triangleright}^*+r_{\triangleleft}^*),V^*)$ is a module of
$(A,\triangleright,\triangleleft)$.
 \end{prop}

\begin{proof}
It is straightforward.
\end{proof}

%\subsection{$\mathcal
%{O}$-operators of L-dendriform algebras and $LD$-equation}

\begin{defn} {\rm Let $(A, \triangleright, \triangleleft)$ be an L-dendriform
algebra and $(l_\triangleright, r_\triangleright, l_\triangleleft,
r_\triangleleft, V)$ be a module. A linear map $T: V \rightarrow A$
is called an {\it $\mathcal {O}$-operator of $(A, \triangleright,
\triangleleft)$ associated to $(l_\triangleright, r_\triangleright,
l_\triangleleft, r_\triangleleft, V)$} if $T$ satisfies (for any
$u,v\in V$)
\begin{equation}
T(u) \triangleright T(v) =T[l_\triangleright(T(u))v +
r_\triangleright(T(v))u],
T(u) \triangleleft T(v) = T[l_\triangleleft(T(u))v +
r_\triangleleft(T(v))u].
\end{equation}
}\end{defn}

By a similar proof as of Theorem 2.12, we obtain the following two
conclusions:

\begin{prop} Let $(A, \triangleright, \triangleleft)$ be an L-dendriform
algebra and $r\in A \otimes A$ be skew-symmetric.  Let $(A,\circ)$
and $(A,\bullet)$ be the associated vertical and horizontal pre-Lie
algebras respectively. Then the following conditions are equivalent:

{\rm (1)} $r$ is an $\mathcal {O}$-operator of $(A, \circ)$
associated to
 $(L_{\triangleright}^* +
L_{\triangleleft}^*,L_{\triangleleft}^*,\ A^*)$.

{\rm (2)} $r$ is an $\mathcal {O}$-operator of  $(A, \bullet)$
associated to  $(L_{\triangleright}^* - R_{\triangleleft}^*,\
-R_{\triangleleft}^*,A^*)$.

{\rm (3)} $r$ satisfies
\begin{equation}
r_{13} \circ r_{23} + r_{12} \bullet r_{23} - r_{12} \triangleleft
r_{13} = 0.
\end{equation}

\end{prop}

\begin{prop} Let $(A,\triangleright,\triangleleft)$
be an L-dendriform algebra and $r \in A \otimes A$ be
skew-symmetric. Then $r$ is an $\mathcal {O}$-operator of
 $(A,\triangleright,\triangleleft)$ associated to
 $(L_{\triangleright}^* + L_{\triangleleft}^* -
R_{\triangleright}^* - R_{\triangleleft}^*,R_{\triangleright}^*,\
R_{\triangleright}^* - L_{\triangleleft}^*,-(R_{\triangleright}^* +
R_{\triangleleft}^*),A^*)$ if and only if $r$ satisfies the
following equations:
\begin{equation}
r_{13} \triangleright r_{23} = -[r_{12}, r_{23}] + r_{13}
\triangleright r_{12},
\end{equation}
\begin{equation}
r_{23} \triangleleft r_{13} = r_{13} \circ r_{12} + r_{23} \bullet
r_{12}.
\end{equation}
\end{prop}

\begin{lemma}
Let $(A,\triangleright,\triangleleft)$ be an L-dendriform algebra
and $r\in A\otimes A$ be skew-symmetric. Then the following
conditions are equivalent:

{\rm (1)} $r$ satisfies Eq. (4.8);

{\rm (2)} $r$ satisfies Eq. (4.10);

{\rm (3)} $r$ satisfies one of the following equations:
\begin{equation}
r_{23} \circ r_{13} - r_{12} \bullet r_{13} + r_{12} \triangleleft
r_{23} = 0;
\end{equation}
\begin{equation}
r_{23} \circ r_{12} + r_{13} \bullet r_{12} + r_{13} \triangleleft
r_{23} = 0;
\end{equation}
\begin{equation}
r_{12} \circ r_{23} + r_{13} \bullet r_{23} + r_{13} \triangleleft
r_{12} = 0;
\end{equation}
\begin{equation}
r_{12} \circ r_{13} - r_{23} \bullet r_{13} + r_{23} \triangleleft
r_{12} = 0.
\end{equation}
\end{lemma}

\begin{proof} Let $\sigma$ be any element in the permutation group
$\Sigma_3$ acting on $\{ 1,2,3\}$. Then $\sigma$ induces a linear
map from $A\otimes A\otimes A$ to $A\otimes A\otimes A$ by (we still
denote it by $\sigma$)
$$\sigma (x_1\otimes x_2\otimes x_3)=x_{\sigma(1)}\otimes
x_{\sigma(2)}\otimes x_{\sigma(3)},\;\;\forall x_1,x_2,x_3\in A.$$
Hence it is straightforward to check that Eqs. (4.10)-(4.14) are the
Eq. (4.8) under the action of the non-unit elements $\sigma\in
\Sigma_3$ respectively. Note that the skew-symmetry of $r$ is used.
\end{proof}

\begin{coro}  Let $(A,\triangleright,\triangleleft)$
be an L-dendriform algebra and $r \in A \otimes A$ be
skew-symmetric. Then Eq. (4.9) holds if Eq. (4.10) holds.
\end{coro}

\begin{proof}
Note that Eq. (4.9) is exactly the difference between Eqs. (4.12)
and (4.13). Then the conclusion follows immediately.
\end{proof}

\begin{defn}{\rm  Let $(A, \triangleright, \triangleleft)$ be an L-dendriform
algebra and $r\in A \otimes A$. Equation (4.8) is called {\it
$LD$-equation} in $(A, \triangleright, \triangleleft)$.}\end{defn}

Combining Proposition 4.4, Proposition 4.5, Lemma 4.6 and
Corollary 4.7, we obtain the following conclusion:

\begin{coro}
Let $(A,\triangleright,\triangleleft)$ be an L-dendriform algebra
and $r \in A \otimes A$ be skew-symmetric.  Let $(A,\circ)$ and
$(A,\bullet)$ be the associated vertical and horizontal pre-Lie
algebras respectively. Then the following conditions are equivalent.

{\rm (1)} $r$ is a solution of $LD$-equation in $(A,\
\triangleright,\triangleleft)$.

{\rm (2)} $r$ is an $\mathcal {O}$-operator of
 $(A,\triangleright,\triangleleft)$ associated to
 $(L_{\triangleright}^* + L_{\triangleleft}^* -
R_{\triangleright}^* - R_{\triangleleft}^*,R_{\triangleright}^*,\
R_{\triangleright}^* - L_{\triangleleft}^*,-(R_{\triangleright}^* +
R_{\triangleleft}^*),A^*)$.

{\rm (3)} $r$ is an $\mathcal {O}$-operator of $(A, \circ)$
associated to
 $(L_{\triangleright}^* +
L_{\triangleleft}^*,L_{\triangleleft}^*,\ A^*)$.

{\rm (4)} $r$ is an $\mathcal {O}$-operator of   $(A, \bullet)$
associated to  $(L_{\triangleright}^* - R_{\triangleleft}^*,\
-R_{\triangleleft}^*,A^*)$.
\end{coro}

\begin{remark}
{\rm Due to the above result, it is reasonable to regard the
$LD$-equation in an L-dendriform algebra as an analogue of the CYBE
in a Lie algebra (also see \cite{Bai1,Bai2} and Theorem 2.12).}
\end{remark}

\begin{lemma}
Let $(A, \triangleright, \triangleleft)$ be an L-dendriform algebra
and $T:A^*\rightarrow A$ be an invertible linear map. Then $T$ is an
$\mathcal{O}$-operator of $(A, \triangleright, \triangleleft)$
associated to $(L_{\triangleright}^* + L_{\triangleleft}^* -
R_{\triangleright}^* - R_{\triangleleft}^*,\
R_{\triangleright}^*,R_{\triangleright}^* - L_{\triangleleft}^*,\
 -(R_{\triangleright}^* +
R_{\triangleleft}^*),A^*)$ if and only if the bilinear form
$\mathcal B$ induced by $T$ through Eq. (1.8) satisfies
\begin{equation}
\mathcal {B}(x \triangleright y, z) =  -\mathcal {B}(y, [x, z])
 - \mathcal {B}(x, z \triangleright y),
\end{equation}
\begin{equation}
\mathcal {B}(x \triangleleft y, z) =  -\mathcal {B}(y, z \circ x) +
\mathcal {B}(x, z \bullet y), \;\;\forall
x, y, z \in A.
\end{equation}
where $x\circ y = x \triangleright y - y \triangleleft x,x\bullet y
= x \triangleright y + x \triangleleft y,[x, y] = x\circ y - y \circ
x = x\bullet y - y\bullet x$. If, in addition, $\mathcal B$ is
skew-symmetric, then $\mathcal B$ satisfies Eq. (4.15) if $\mathcal
B$ satisfies Eq. (4.16).

\end{lemma}

\begin{proof} It is straightforward.\end{proof}

\begin{coro}
Let $(A, \triangleright, \triangleleft)$ be an L-dendriform algebra
and $r \in A\otimes A$. Suppose that $r$ is skew-symmetric and
invertible. Then $r$ is a solution of $LD$-equation in $(A,
\triangleright, \triangleleft)$ if and only if the nondegenerate
bilinear form $\mathcal {B}$ induced by $r$ through Eq. (1.8)
satisfies Eq. (4.16).
\end{coro}

\begin{proof}
It follows from Lemma  4.11 and Corollary 4.9.
\end{proof}

\begin{defn}
{\rm Let $(A, \triangleright, \triangleleft)$ be an L-dendriform
algebra. A skew-symmetric bilinear form $\mathcal {B}$ satisfying
Eq. (4.16) is called a {\it 2-cocycle} of  $(A, \triangleright,
\triangleleft)$.}
\end{defn}

%\begin{coro} Let $(A, \triangleright,
%\triangleleft)$ be an L-dendriform algebra and $r \in A\otimes A$ be
%a skew-symmetric solution of $LD$-equation in $A$. Then there exist
%two L-dendriform algebra structures on $A^*$ given by
%\begin{equation}
%a^* \triangleright b^* = (L_{\triangleright}^* +
%L_{\triangleleft}^*)(r(a^*))b^*,a^* \triangleleft b^* = -
%L_{\triangleleft}^*(r(a^*))b^*, \forall a^*, b^* \in A^*,
%\end{equation}
%and
%\begin{equation}
%a^* \triangleright b^* = (L_{\triangleright}^* -
%R_{\triangleleft}^*)(r(a^*))b^*,a^* \triangleleft b^* =
%R_{\triangleleft}^*(r(a^*))b^*, \forall a^*, b^* \in A^*,
%\end{equation}
%respectively. Their transposes are given respectively by
%\begin{equation}
%a^* \triangleright^t b^* = (L_{\triangleright}^* +
%L_{\triangleleft}^*)(r(a^*))b^*,a^* \triangleleft^t b^* =
%L_{\triangleleft}^*(r(b^*))a^*, \forall a^*, b^* \in A^*,
%\end{equation}
%and
%\begin{equation}
%a^* \triangleright^t b^* = (L_{\triangleright}^* -
%R_{\triangleleft}^*)(r(a^*))b^*,a^* \triangleleft^t b^* =
%-R_{\triangleleft}^*(r(b^*))a^*, \forall a^*, b^* \in
%A^*.\end{equation}
%\end{coro}

%\begin{proof}
%It follows from  Theorems 3.1.9, Proposition 4.2.3 and Proposition
%3.1.6.
%\end{proof}

 By a similar proof as of Theorem 2.13. we have the following conclusion:

\begin{theorem} Let $(A, \triangleright, \triangleleft)$
be an L-dendriform algebra and $(l_\triangleright, r_\triangleright,
l_\triangleleft, r_\triangleleft, V)$ be a module.  Let $T:
V\rightarrow A$ be a linear map which can be identified as an
element in the vector space $(A\oplus V^*)\otimes (A\oplus V^*)$.
Then $ r = T - \sigma(T)$  is a skew-symmetric solution of
$LD$-equation in the L-dendriform algebra
$A\ltimes_{l_{\triangleright}^*+l_{\triangleleft}^*-r_{\triangleright}^*-r_{\triangleleft}^*,\
r_{\triangleright}^*,r_{\triangleright}^*-l_{\triangleleft}^*,\
-(r_{\triangleright}^*+r_{\triangleleft}^*)}V^*$ if and only if $T$
is an $\mathcal{O}$-operator of  $(A, \triangleright,
\triangleleft)$ associated to $(l_\triangleright, r_\triangleright,
l_\triangleleft, r_\triangleleft, V)$.
\end{theorem}

\section{Generalization}
 \setcounter{equation}{0}
\renewcommand{\theequation}
{5.\arabic{equation}}

The study in the previous sections motivates us to consider the
following structures: let $(X,[,])$ be a Lie algebra and
$(*_i)_{1\leq i\leq N}:X\otimes X\rightarrow X$ be a family of
binary operations on $X$. Then the Lie bracket $[,]$ splits into the
commutator of $N$ binary operations $*_1,\cdots,
*_N$ if
\begin{equation}
[x,y]=\sum_{i=1}^N (x*_iy-y*_ix),\;\;\forall x,y\in X.
\end{equation} Note that when $N=1$, the algebra with the binary operation
$*_1=*$ is a Lie-admissible algebra.

Like the Loday algebras, only Eq. (5.1) is too general to get more
interesting structures. So some additional conditions for the binary
operations $*_i$ are necessary. We pay our main attention to the
case that $N=2^n$, $n=0,1,2,\cdots$. As the study in the cases of
associative algebras given in \cite{Bai3} and the study in the
previous sections, we define a ``rule'' of constructing the binary
operations $*_i$ as follows: the $2^{n+1}$ binary operations give a
natural module structure of an algebra with the $2^{n}$ binary
operations on the underlying vector space of the algebra itself,
which is the beauty of such algebra structures.  That is, by
induction, for the algebra $(A,
*_i)_{1\leq i\leq 2^n}$, besides the natural module of $A$ on
 the underlying vector space of
$A$ itself given by the left and right multiplication operators, one
can introduce the $2^{n+1}$ binary operations $\{*_{i_1},
*_{i_2}\}_{1\leq i\leq 2^n}$ such that
\begin{equation}
x*_i y=x*_{i_1} y-y*_{i_2}x,\;\;\forall x,y\in A, \; 1\leq i\leq
2^n, \end{equation} and their left or right multiplication operators
give a module of $(A,
*_i)_{1\leq i\leq 2^n}$ by acting on the underlying vector space of
$A$ itself.

In particular, when $N=1$ and $N=2$, the corresponding algebra
$(A,*_i)_{1\leq i\leq N}$ according to the above rule is exactly a
pre-Lie algebra and an L-dendriform algebra respectively. On the
other hand, note that for $n\geq 1$ ($N\geq 2$), in order to make
Eq. (5.1) be satisfied, there is an alternative (sum) form of Eq.
(5.2)
\begin{equation} x*_i y=x*_{i_1}y +x*_{i_2}'y,\;\;\forall x,y\in
A,\; 1\leq i\leq 2^n,\end{equation} by letting $x
*_{i_2}'y=-y*_{i_2}x$ for any $x,y\in A$. In particular, in such a
situation, it can be regarded as a binary operation $*$ of a pre-Lie
algebra that splits into the $N=2^n$ ($n=1,2\cdots$) binary
operations $*_1,...,*_N$. In this sense, an L-dendriform algebra is
also regarded as a ``pre-Lie algebraic analogue" of a dendriform
algebra.

We would like to point out that in the case of associative algebras,
there is an outline of a study by induction on the algebra systems
with more binary operations given in \cite{Bai3}, which is still
valid for the case of Lie algebras.

\section*{Acknowledgements}

This work was supported in part by NSFC (10621101, 10921061), NKBRPC
(2006CB805905) and SRFDP (200800550015). We thank the referee for
the important suggestions.

\end{document}